\documentclass[runningheads]{llncs}
\usepackage[a4paper,margin=3cm]{geometry}
\usepackage[utf8]{inputenc}
\usepackage{graphicx}
\usepackage{subcaption}
\usepackage{tikz}
\newcommand{\tikzcircle}[2][red,fill=red]{\tikz[baseline=-0.5ex]\draw[#1,radius=#2] (0,0) circle ;}%

\usepackage{footnote}
\makesavenoteenv{tabular}
\usepackage[symbol]{footmisc}

\usepackage{amsmath,amssymb,amsthm}

\usepackage{booktabs}
\usepackage{enumerate,paralist}
\usepackage[basic]{complexity}

\newtheorem{observation}[definition]{Observation}
\newtheorem{prp}[theorem]{Property}

\renewcommand{\orcidID}[1]{\href{https://orcid.org/#1}{\includegraphics[scale=.03]{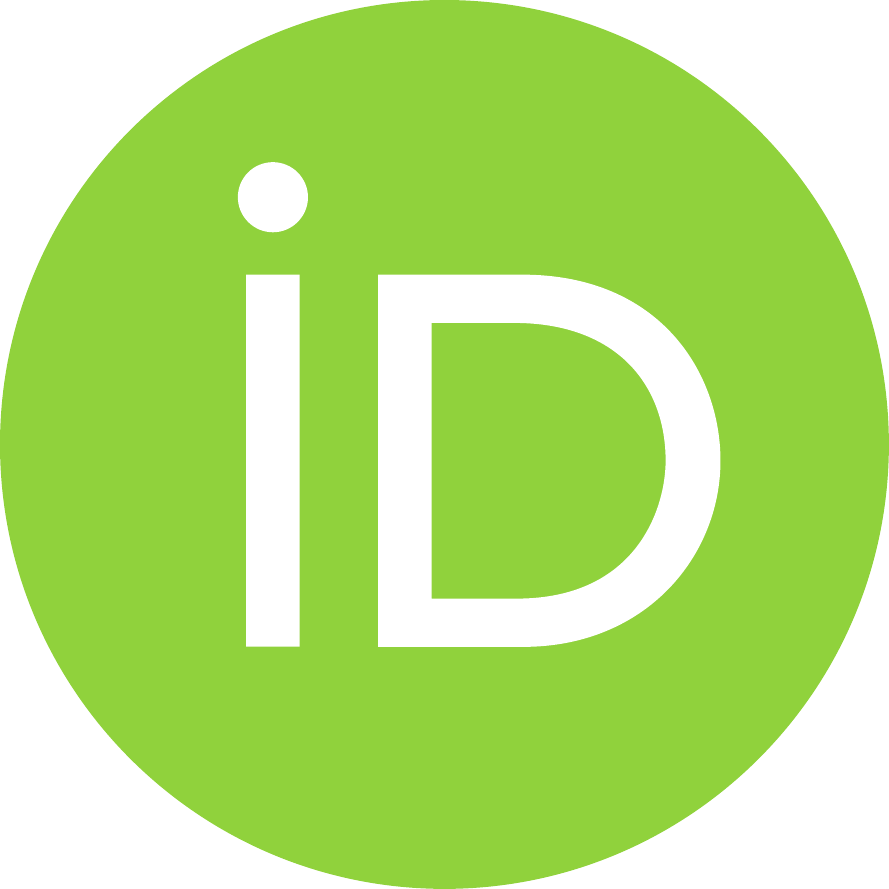}}}
  
\usepackage[nocompress,noadjust]{cite}
\usepackage{xcolor}
\usepackage[pdfpagelabels,colorlinks,urlcolor=blue,linkcolor=blue,citecolor=blue]{hyperref}
\usepackage[nameinlink,capitalize]{cleveref}

\usepackage{todonotes}

\graphicspath{{figures/}}

\newcommand{\Oh}{{\ensuremath{\mathcal{O}}}}
\newcommand{\Sh}{{\ensuremath{\mathcal{S}}}}
\newcommand{\Qh}{{\ensuremath{\mathcal{Q}}}}
\newcommand{\grid}{\ensuremath{H}}

\DeclareMathOperator{\sn}{sn}
\DeclareMathOperator{\qn}{qn}

\title{The Mixed Page Number of Graphs}
\titlerunning{The Mixed Page Number of Graphs}
\author{Jawaherul~Md.~Alam\inst{1}\orcidID{0000-0002-7062-6792} \and 
Michael~A.~Bekos\inst{2}\orcidID{0000-0002-3414-7444} \and 
Martin~Gronemann\inst{3}\orcidID{0000-0003-2565-090X} \and\newline
Michael~Kaufmann\inst{2}\orcidID{0000-0001-9186-3538} \and 
Sergey~Pupyrev\inst{4}\orcidID{0000-0003-4089-673X}}
\authorrunning{Alam et al.}

\institute{
Amazon Inc., Tempe, AZ, USA
\\\email{jawaherul@gmail.com}
\and
Institut f{\"u}r Informatik, Universit{\"a}t T{\"u}bingen, T{\"u}bingen, Germany
\\\email{\{bekos,mk\}@informatik.uni-tuebingen.de}
\and
Theoretical Computer Science, Osnabr\"uck University, Osnabr\"uck, Germany
\\\email{martin.gronemann@uni-osnabrueck.de}
\and
Facebook, Inc., Menlo Park, CA, USA
\\\email{spupyrev@gmail.com}
}

\begin{document}
\maketitle

\begin{abstract}
A linear layout of a graph typically consists of a total vertex order, and a partition of the edges into sets of either non-crossing edges, called stacks, or non-nested edges, called queues. The stack (queue) number of a graph is the minimum number of required stacks (queues) in a linear layout. Mixed linear layouts combine these layouts by allowing each set of edges to form either a stack or a queue.
In this work we initiate the study of the \emph{mixed page number} of a graph which corresponds to the minimum number of such sets.

First, we study the edge density of graphs with bounded mixed page number. Then, we focus on complete and complete bipartite graphs, for which we derive lower and upper bounds on their mixed page number. Our findings indicate that combining stacks and queues is more powerful in various ways compared to the two traditional layouts.  
\end{abstract}

\section{Introduction}
Linear layouts of graphs form an important research topic, which has been studied in different areas and under different perspectives. As a matter of fact, several combinatorial optimization problems are defined by means of a measure over a linear layout of a graph  (including the well-known cutwidth~\cite{doi:10.1137/0125042}, bandwidth~\cite{DBLP:journals/jgt/ChinnCDG82} and pathwidth~\cite{DBLP:journals/jct/RobertsonS83}). As a result, the corresponding literature is rather rich;~see~\cite{DBLP:journals/eatcs/SernaT05}. Typically, a linear layout of a graph is defined by a total order of its vertices together with an objective over its edges that one seeks to optimize. 

In this work, we focus on linear layouts in which the edges must be partitioned into a minimum number of sets, called \emph{pages}, such that each page in the partition has a certain property~\cite{DBLP:journals/dmtcs/DujmovicW04}. The most prominent representatives in this category are the \emph{stack layouts} (also known as \emph{book embeddings}) and the \emph{queue layouts}. The former do not allow two edges in the same page (called \emph{stack}, in this context)  to cross~\cite{DBLP:journals/jct/BernhartK79}, while in the latter no two edges of the same page (called \emph{queue}, in~this context) can nest~\cite{DBLP:journals/siamdm/HeathLR92}; see~\cref{fig:layouts-samples}. In other words, the endpoints~of~the edges of the same stack (queue) follow the LIFO (FIFO) principle in the underlying~order. 

Given a graph, its \emph{stack} (\emph{queue}) \emph{number}  is defined as the minimum number of stacks (queues) in an stack (queue) layout of the graph. Both of these graph parameters have been extensively studied; see~\cite{DBLP:journals/jacm/DujmovicJMMUW20,DBLP:journals/dam/GanleyH01,DBLP:journals/siamdm/HeathLR92,DBLP:journals/combinatorics/Wiechert17,DBLP:journals/jcss/Yannakakis89}. For example, it is known that both the stack  and the queue number of the $n$-vertex complete graph, $K_n$, is $\lfloor \frac{n}{2} \rfloor$~\cite{DBLP:journals/jct/BernhartK79,DBLP:journals/siamcomp/HeathR92}. However, for several families of graphs, the exact stack  or queue number is unknown. This motivates the study of corresponding upper and lower bounds for these parameters. The most notable example in this category~regards the stack number of the complete bipartite graph $K_{n,n}$; the best-known lower bound is  $\lfloor \frac{n}{2} \rfloor$~\cite{DBLP:journals/jct/BernhartK79}, while the corresponding best upper bound is $\lfloor \frac{2n}{3} \rfloor + 1$~\cite{DBLP:journals/jct/EnomotoNO97}. On the other hand, the queue number of $K_{n,n}$ is exactly $\lfloor \frac{n}{2} \rfloor$~\cite{DBLP:journals/siamcomp/HeathR92}. Another notable example is the queue number of planar graphs; the best-known lower bound is  $4$~\cite{DBLP:journals/algorithmica/AlamBGKP20}, while the corresponding best upper bound is $49$~\cite{DBLP:journals/jacm/DujmovicJMMUW20}. Note that the stack number of the planar graphs was recently shown to be exactly~$4$~\cite{DBLP:journals/jocg/KaufmannBKPRU20,DBLP:journals/jcss/Yannakakis89}.

\begin{figure}[t!]
	\centering	
	\begin{subfigure}[b]{.32\textwidth}
		\centering
		\includegraphics[scale=1,page=1]{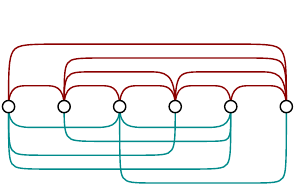}
		\caption{}
		\label{fig:mixed-sample}
	\end{subfigure}
	\hfil
	\begin{subfigure}[b]{.32\textwidth}
		\centering
		\includegraphics[scale=1,page=3]{figures/K6intro}
		\caption{}
		\label{fig:stack-sample}
	\end{subfigure}
	\hfil
	\begin{subfigure}[b]{.32\textwidth}
		\centering
		\includegraphics[scale=1,page=2]{figures/K6intro}
		\caption{}
		\label{fig:queue-sample}
	\end{subfigure}
	\caption{Different layouts of $K_6$: 
	(a)~$1$-stack $1$-queue, 
	(b)~$3$-stack, and 
	(c)~$3$-queue.}
	\label{fig:layouts-samples}
\end{figure}

A natural generalization of stack and queue layouts was introduced back in 1992 by Heath and Rosenberg~\cite{DBLP:journals/siamcomp/HeathR92}, who proposed the study of $s$-stack $q$-queue layouts that consist of $s$ stacks and $q$ queues. In the seminal paper~\cite{DBLP:journals/siamcomp/HeathR92}, they conjectured that every planar graph admits a $1$-stack $1$-queue layout. This conjecture, however, was settled in the negative by Pupyrev~\cite{DBLP:conf/gd/Pupyrev17} in 2017, while more recently Bekos et al.~\cite{DBLP:conf/gd/AngeliniBKM20} showed that the conjecture does not hold even for series-parallel graphs. 
Dujmovi{\'c} and Wood~\cite{DBLP:journals/dmtcs/DujmovicW05}, and Enomoto and Miyauchi~\cite{EM14} showed that a graph $G$ admits a $1$-stack $1$-queue linear layout, when each edge can be subdivided $\Oh(\log \sn)$ or $\Oh(\log \qn)$ times, where $\sn$ and $\qn$ denote the stack and queue numbers of $G$, respectively. For the case of planar graphs, Pupyrev~\cite{DBLP:conf/gd/Pupyrev17} showed that one subdivision vertex per edge is sufficient to guarantee a $1$-stack $1$-queue layout.
Even though the problem of studying mixed linear layouts has been highlighted in several works~\cite{DBLP:journals/siamcomp/BekosFGMMRU19,DBLP:journals/jacm/DujmovicJMMUW20,DBLP:journals/jocg/KaufmannBKPRU20}, the corresponding literature is still limited and ends with another hardness result by de~Col et al.~\cite{DBLP:conf/gd/ColKN19}, who proved that determining whether a graph admits a $2$-stack $1$-queue layout~is~\NP-complete.

Despite these mostly-negative results, one still naturally expects a reduction on the total number of pages (i.e., stacks or queues) required in a mixed layout with respect to the corresponding ones required in pure stack or queue layouts. This expectation, however, has not been confirmed so far, mainly due to the different nature of stacks and queues, which makes them difficult~to~be~combined.    

\medskip\noindent\textbf{Our contribution.} In this work, we confirm the aforementioned expectation. To achieve this, we introduce and study a new graph parameter, called \emph{mixed page number}, which equals to the minimum value of $s+q$ for which an $s$-stack $q$-queue layout for the graph exists. Our contribution is as follows:
\begin{itemize}[--]
\item In \cref{sec:properties}, we study properties of graphs admitting $s$-stack $q$-queue layouts, including their edge density and the complexity of their recognition. 
\item In \cref{sec:complete-graphs}, we show that the mixed page number of the $n$-vertex complete graph $K_n$ is at least $\lceil\frac{3(n-4)}{8}\rceil$ and at most $2 \lceil \frac{n}{5} \rceil$, which implies that its mixed page number is strictly less than its stack and queue number (e.g.,~the stack and queue number of $K_6$ is $3$, while its mixed page number is $2$; see~\cref{fig:layouts-samples}).
\item In \cref{sec:complete-bipartite-graphs}, we focus on the complete bipartite graph $K_{n,n}$. We prove that its mixed page number is  $\lceil \frac{2n}{3} \rceil$ in the special case, in which~all vertices of one of its parts precede those of the other. For the general setting, we show that the mixed page number of $K_{n,n}$ ranges between $\lceil \frac{n}{3} \rceil$~and~$\lfloor\frac{n}{2}\rfloor$.
\end{itemize}
We present preliminary notions and definitions in \cref{sec:preliminaries}. In \cref{sec:conclusions}, we conclude with a list of interesting open problems raised by our work.

\section{Preliminary Definitions and Notation}
\label{sec:preliminaries}
A \emph{vertex order} $\prec$ of a graph $G$ is a total order of its vertices, such that for any two vertices $u$ and $v$ of $G$, $u \prec v$ if and only if $u$ precedes $v$ in the order. We write $[u_1,\ldots, u_k ]$ to denote $u_i \prec u_{i+1}$ for all $1 \leq i \leq k-1$. Let $F$ be a set of $k \geq 2$ independent edges $(u_i, v_i)$ of $G$, that is, $F=\{(u_i, v_i);\;i=1,\ldots,k\}$. If the order is $[u_1, \ldots, u_k, v_k, \ldots, v_1]$, then we say that the edges of $F$ form a \emph{$k$-rainbow}, while if the order is $[u_1, v_1, \ldots, u_k, v_k]$, then the edges of $F$ form a \emph{$k$-necklace}. The edges of $F$ form a \emph{$k$-twist}, if the order is $[u_1, \ldots, u_k, v_1, \ldots, v_k]$; see \cref{fig:necklace-twist}. Two independent edges that form a $2$-twist ($2$-rainbow, $2$-necklace, respectively) are commonly referred to as \emph{crossing} (\emph{nested}, \emph{disjoint}, respectively). A stack (queue) is a set of pairwise non-crossings (non-nested, respectively) edges in $\prec$.

\begin{figure}[t!]
	\centering	
	\begin{subfigure}[b]{.32\textwidth}
		\centering
		\includegraphics[scale=1,page=1]{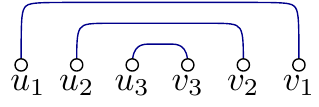}
		\caption{}
		\label{fig:rainbow}
	\end{subfigure}
	\hfil
	\begin{subfigure}[b]{.32\textwidth}
		\centering
		\includegraphics[scale=1,page=3]{figures/preliminaries}
		\caption{}
		\label{fig:necklace}
	\end{subfigure}
	\hfil
	\begin{subfigure}[b]{.32\textwidth}
		\centering
		\includegraphics[scale=1,page=2]{figures/preliminaries}
		\caption{}
		\label{fig:twist}
	\end{subfigure}
	\caption{Illustration of:
		(a)~a $3$-rainbow, 
		(b)~a $3$-necklace, and 
		(c)~a $3$-twist.}
	\label{fig:necklace-twist}
\end{figure}

A mixed \emph{$s$-stack} \emph{$q$-queue} \emph{layout} of a graph consists of a vertex order~$\prec$, called \emph{linear order}, and a partition of its edges into $s$ stacks and $q$ queues. Hence, an $s$-stack ($q$-queue) layout is a mixed $s$-stack $0$-queue ($0$-stack $q$-queue) layout.~The \emph{stack} (\emph{queue}) \emph{number} of a graph $G$ is the minimum $k$, such that $G$ admits~a~$k$-stack ($k$-queue) layout. In this work, we study the mixed page number~of~a~graph. 

\begin{definition}[Mixed page number]\label{def:mixed-page-number}
The \emph{mixed page number} of a graph $G$ is the minimum value of $s+q$ such that $G$ admits an $s$-stack $q$-queue layout.
\end{definition}

\noindent Note that in a $k$-page mixed layout, each of the $k$ pages can be either a stack {\it or} a queue.
We stress the difference with so-called simultaneous layouts in which every page is a stack {\it and} a queue~\cite{DBLP:journals/corr/abs-2007-15102}.
The next property follows from \cref{def:mixed-page-number}.

\begin{prp}\label{prp:numbers}
The mixed page number of a graph is upper bounded by its stack number and by its queue number.
\end{prp}

\section{Basic Properties and Preliminary Results}
\label{sec:properties}

In this section, we introduce basic properties of mixed linear layouts. We start with a somehow unexpected result, which is positioned within the literature as follows. For a graph with $m$ edges, it is known that its stack and its queue numbers are upper bounded by $31\sqrt{m}$~\cite{DBLP:journals/jal/Malitz94} and $e\sqrt{m}$~\cite{DBLP:journals/dmtcs/DujmovicW04}, respectively. Malitz~\cite{DBLP:journals/jal/Malitz94}, and Dujmov{\`i}c and Wood~\cite{DBLP:journals/dmtcs/DujmovicW04} proved these bounds by carefully choosing the underlying linear order (as a certain linear order may contain an $\lfloor \frac{n}{2} \rfloor$-twist and an $\lfloor \frac{n}{2} \rfloor$-rainbow, whose presence does not imply any valuable bound on the stack or on the queue number). In the following, we prove a corresponding upper bound on the mixed number, which notably holds for every fixed linear order~of~the~vertices.

\begin{theorem}\label{thm:m-edges}
The mixed page number of a graph with $m$ edges is at most $\lfloor \sqrt{2m} \rfloor$ for every fixed linear order of the vertices.
\end{theorem}
\begin{proof}
Let $G$ be a graph with $m$ edges. Consider the maximum rainbow formed in some fixed linear order of its vertices. If its size is at most $\sqrt{2m}$, then all edges of $G$ can be assigned to at most $\lfloor \sqrt{2m} \rfloor$ queues~\cite{DBLP:journals/siamcomp/HeathR92}. Otherwise, we assign~the~edges of the largest rainbow to a single stack, and iteratively apply the argument to the remaining graph. 
The number of iterations $T(m)$ corresponding to the total number of stacks and queues used, is bounded by the following recurrence:
%
\[
T(m) \leq \begin{cases} 
1 &\mbox{if } m=1 \\
T(m - \lceil \sqrt{2m}\rceil) + 1   & \mbox{otherwise } 
\end{cases}
\]
%
Using strong induction and the fact that $m - \sqrt{2m} = (\sqrt{m} - 1/\sqrt{2})^2 - 1/2$, it can be easily proven that $T(m) \leq \lfloor \sqrt{2m} \rfloor$, which completes the proof.
\end{proof}

\noindent In contrast to the positive result above that guarantees an upper bound on~the mixed page number of a graph when the vertex order is fixed, in the same setting, determining the minimum number of required pages turns out to be \NP-hard.

\clearpage
\begin{theorem}\label{thm:fixed_hardness}
Given a graph with a fixed vertex order, it is \NP-hard to determine its mixed page number respecting the vertex order.
\end{theorem}
\begin{proof}
Our reduction is from the \NP-complete problem of deciding whether a permutation is \emph{$k$-coverable}~\cite{DBLP:journals/eik/Wagner84}, i.e., whether it can be partitioned into $k$ monotone subsequences. Given a permutation $\pi=\langle \pi_1, \ldots, \pi_n\rangle$ of size $n$, we construct a graph $G$ on $2n$ vertices, whose vertex order is $1, \dots, n, \pi_1, \dots, \pi_n$, and whose edges are $(1, \pi_1),\ldots, (n,\pi_n)$, i.e., $G$ is a matching.
Then, a stack in $G$ corresponds to a decreasing subsequence of $\pi$, while a queue 
is an increasing subsequence. Thus, finding a mixed $s$-stack $q$-queue layout with 
$s+q=k$ for $G$ under the vertex order is equivalent to deciding if $\pi$ is $k$-coverable, proving the claim.
\end{proof}	

\noindent Note that in order to determine whether a graph admits a mixed layout with at most two pages in a certain vertex order in polynomial time, we can use a simple $2$-SAT formulation, in which each edge $e$ yields a variable $x_e$ such that $x_e = \texttt{true}$ if and only if $e$ is assigned to the first of the two available pages. Then, crossings or nestings can be easily formulated by $2$-SAT clauses; see~\cite{DBLP:conf/gd/Bekos0Z15}.

\begin{observation}
Given a graph with a fixed vertex order, we can decide in polynomial time whether its mixed page number respecting the vertex order is at~most~$2$.
\end{observation}

\noindent Since there exist permutations with length $n$ that cannot be decomposed in $o(\sqrt{n})$ monotone subsequences~\cite{DBLP:journals/eik/BrandstadtK86}, the construction given in the proof of \cref{thm:fixed_hardness} implies the existence of graphs with $m$ edges and corresponding fixed linear vertex-orders, whose mixed page number is $\Omega(\sqrt{m})$ respecting the vertex order. In the following theorem, we continue our study on properties of mixed linear layouts by answering a so-called extremal question regarding the maximum number of edges contained in a mixed linear layout of a graph. 

\begin{theorem}\label{thm:density}
Every $n$-vertex graph with mixed page number $k$ has at most $f(n,k)$ edges, where: 
%
\begin{equation*}
f(n,k) = \begin{cases} 
2kn - 2k^2 + k - 2 &\mbox{if } k \leq \frac{n}{4}+2 \\
\frac{n^2}{8} + (k+1) n - 3 k -2  & \mbox{otherwise } 
\end{cases}
\end{equation*}
%
Also, for every $n$ and $k$ with $n \geq 4k + 1$, there exists an $n$-vertex graph with mixed page number $k$ and $2kn-2k^2 + k - 2$ edges.
\end{theorem}
\begin{proof}
Let $G$ be an $n$-vertex graph with mixed page number $k$ and let $\mathcal{L}$ be a mixed layout of it with $s$ stacks and $q$ queues, such that $s+q=k$.  
It is known that the first stack of $\mathcal{L}$ has at most $2n-3$ edges, while each subsequent stack of $\mathcal{L}$ has at most $n-3$ edges~\cite{DBLP:journals/jct/BernhartK79}. Also, for $1\leq i\leq q$, the $i$-th queue of $\mathcal{L}$ has at most $2n-4i+1$ edges~\cite{DBLP:journals/dmtcs/DujmovicW04}. 
Since we seek for an upper bound, we may assume w.l.o.g.\ that each stack or queue of $\mathcal{L}$ contains the maximum amount of edges it can hold.
In this regard, we may further assume w.l.o.g.\ that the first stack is always part of $\mathcal{L}$, since its contribution to the number of edges of $G$ is maximum. We next argue about the remaining $k-1$ pages of $\mathcal{L}$.
Since the sparsest queue contains $2n -4(k-1) +1$ edges, while all remaining stacks contain $n-3$ edges each, it follows that, if $k \leq \frac{n}{4}+2$, then $\mathcal{L}$ is a $1$-stack, $(k-1)$-queue layout. Otherwise, $\mathcal{L}$ contains  $k-\frac{n}{4}-1$ stacks and $\frac{n}{4}+1$ queues.

Assume first that $k \leq \frac{n}{4}+2$. Let $v_0,v_1,\ldots,v_{n-1}$ be the vertices of $G$ in the order that appear in $\mathcal{L}$.  
Since $\mathcal{L}$ contains one stack, we may assume w.l.o.g.\ that, for all integers $0 \leq i \leq n-2$, $G$ contains edge $(v_i,v_{i+1})$, which belongs to the stack in $\mathcal{L}$. 
So, we next focus on the edges of the $k-1$ queues.
For all integers $0 \leq i < j \leq n-1$, let the \emph{midpoint}~\cite{DBLP:journals/dmtcs/DujmovicW04} of edge $(v_i,v_j)$ of $G$ be $\frac{1}{2}(i+j)$. 
Since two edges with the same midpoint form a $2$-rainbow~\cite{DBLP:journals/dmtcs/DujmovicW04}, they cannot belong to the same queue. 
For all integers $1 \leq i \leq k-1$, at most $i-1$ edges have a midpoint of $i-1$, and at most $i-1$ edges (excluding edges in the stack) have a midpoint of $i-\frac{1}{2}$. 
Also, for all integers $1 \leq i \leq k-1$, at most $i-1$ edges
have a midpoint of $n-i$, and at most $i-1$ edges (again excluding edges in the stack) have a midpoint of $n-i-\frac{1}{2}$. Since $n \geq 2k$, we avoid double-counting. 
Hence, the number of edges of $G$, neglecting the ones in the stack, is at most:
%
\[
4\sum_{i=1}^{k-1} (i-1)  \; +  \;\;
(2n - 1 - 4(k-1))(k-1) \;
\]
%
The latter is equal to $2kn - 2k^2 + k - 2n + 1$. So, $G$ has at most $2kn - 2k^2 + k - 2$ edges.

\begin{figure}[t]
	\begin{subfigure}[b]{.48\linewidth}
		\flushleft
		\includegraphics[page=6,width=\textwidth]{density-lowerbound-2}
		\caption{$\Sh$}
	\end{subfigure}\hfill
	\begin{subfigure}[b]{.48\linewidth}
		\flushright
		\includegraphics[page=1,width=\textwidth]{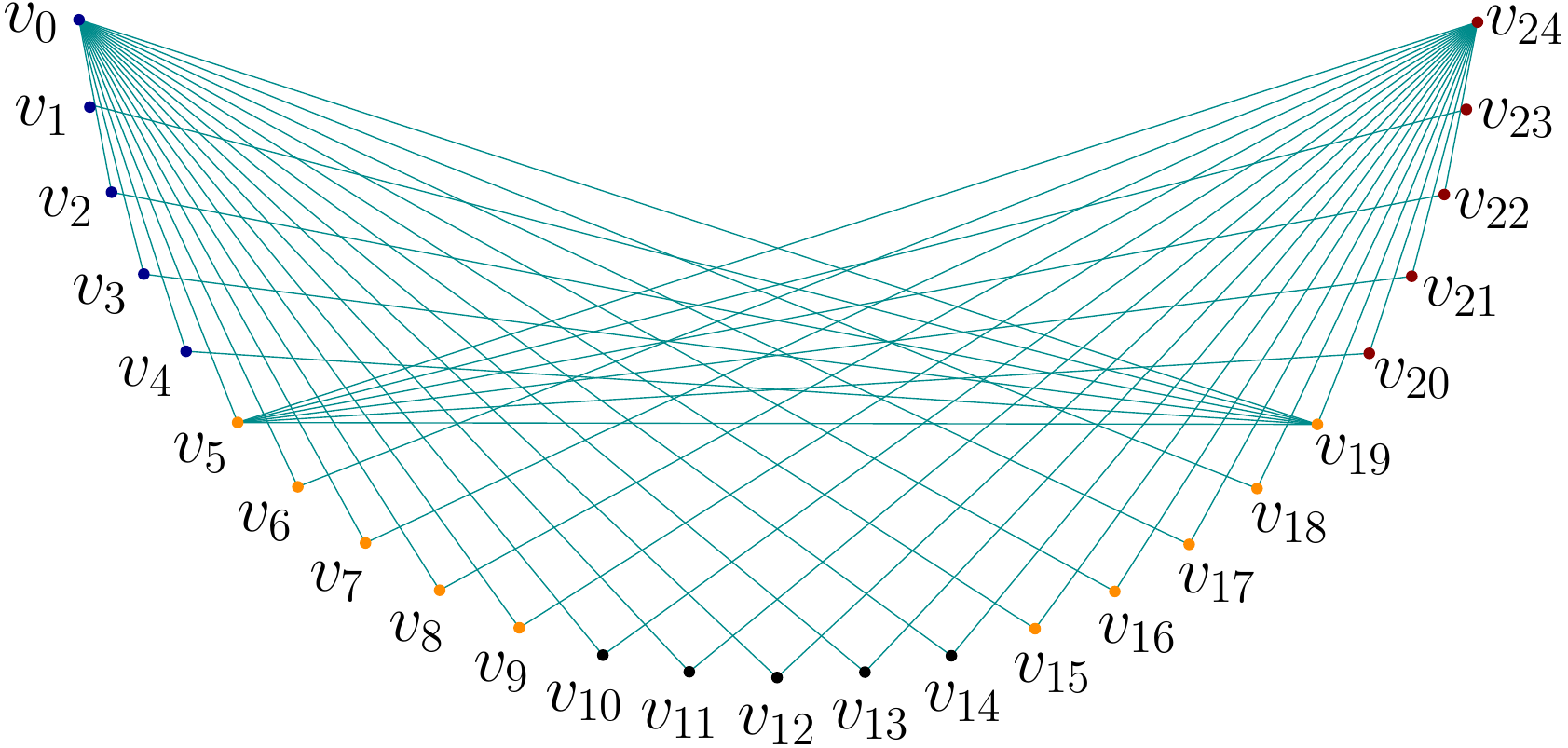}
		\caption{$\Qh_0$}
		\label{fig:lb_q0}
	\end{subfigure}    
	\begin{subfigure}[b]{.48\linewidth}
		\flushleft
		\includegraphics[page=2,width=\textwidth]{density-lowerbound-2}
		\caption{$\Qh_1$}
		\label{fig:lb_q1}
	\end{subfigure}\hfill
	\begin{subfigure}[b]{.48\linewidth}
		\flushright
		\includegraphics[page=3,width=\textwidth]{density-lowerbound-2}
		\caption{$\Qh_2$}
		\label{fig:lb_q2}
	\end{subfigure}
	\begin{subfigure}[b]{.48\linewidth}
		\flushleft
		\includegraphics[page=4,width=\textwidth]{density-lowerbound-2}
		\caption{$\Qh_3$}
		\label{fig:lb_q3}
	\end{subfigure}\hfill
	\begin{subfigure}[b]{.48\linewidth}
		\flushright
		\includegraphics[page=5,width=\textwidth]{density-lowerbound-2}
		\caption{$\Qh_4$}
		\label{fig:lb_q4}
	\end{subfigure}
	\caption{Illustration for the lower bound of $2nk - 2k^2 + k - 2$ of \cref{thm:density} with $n = 25$ and $k = 6$ which yields a $1$-stack $5$-queue layout with $232$ edges in total.}
	\label{fig:mixed-density-lb}
\end{figure}

We now consider the case in which $k > \frac{n}{4}+2$, i.e., $\mathcal{L}$ contains $k-\frac{n}{4}-1$ stacks and $\frac{n}{4}+1$ queues. By our discussion above, $\mathcal{L}$ contains $\frac{3n^2}{8} + \frac{9 n}{4} - 8$ edges in total in its first stack and in its $\frac{n}{4}+1$ queues. Each of the remaining  $k - \frac{n}{4} - 2$ stacks of $\mathcal{L}$ contain at most $n-3$ edges, so in total, $kn - 3 k - \frac{n^2}{4} - \frac{5 n}{4} + 6$. By summing up, we can conclude that the number of edges in $G$ is at most $\frac{n^2}{8} + (k+1) n - 3 k - 2$.

For given integers $n$ and $k$, such that $n \ge 4k+1$, we present an $n$-vertex graph $G_{n,k}$ with $2nk-2k^2+k-2$ edges, which admits a $1$-stack, $(k-1)$-queue layout $\mathcal{L}_{n,k}$. To ease the presentation, we assume that $n$ is odd; the construction~easily generalizes to even values of $n$. Since $n \ge 4k+1$, by our upper bound, the mixed page number of~$G_{n,k}$ cannot be $k-1$. Let $v_0,\ldots,v_{n-1}$ be the vertices of $G_{n,k}$ as they appear in $\mathcal{L}_{n,k}$. 
The first $k-2$ queues of $\mathcal{L}_{n,k}$ are symmetric, namely, for each $0 \leq i \leq k-3$, queue $\Qh_i$ contains $2n-4i-5$ edges (see \cref{fig:lb_q0,fig:lb_q1,fig:lb_q2,fig:lb_q3}):

\medskip 
\begin{tabular}{@{~--~} l @{\hspace{1.5cm}} r @{$ \;\; \leq \; j \; \leq \;\; $} l} 
$(v_i, v_j)$              & $i + 2$  & $n - k - i$ \\
$(v_{k + i - 1}, v_{j})$  & $n-k-i$   & $ n - i - 2$ \\
$(v_{j}, v_{n-k-i})$      & $i+1$     & $ k+i-2$ \\
$(v_{j}, v_{n-i-1})$      & $k + i -1$ & $ n - i - 3$ 
\end{tabular}

\medskip 
The structure of the last queue $\Qh_{k-1}$ of $\mathcal{L}_{n,k}$, however, is different. It contains three so-called double-stars, where a \emph{double-star} rooted at $\langle v_\ell, v_r \rangle$ with $0 \leq \ell < r \leq n-1$ contains the following $2(r-\ell)-3$ edges: $(v_{\ell}, v_{\ell+2}),\ldots,(v_{\ell}, v_{r})$ and $(v_{\ell+1}, v_{r}),\ldots,(v_{r-2}, v_{r})$. With this definition at hand, the three double-stars of $\Qh_{k-1}$ of $\mathcal{L}_{n,k}$ are rooted at $\langle v_{-2 - 2 k + n}, v_{(n-1)/2 - 1} \rangle$, $\langle v_{(n-1)/2 -3}, v_{(n-1)/2 +3} \rangle$, and $\langle v_{(n-1)/2+1}, v_{n-k+1} \rangle$, and in total contain $2n - 4k + 3$ edges; see \cref{fig:lb_q4}.
Summing up over $i$, we obtain that the total number of edges in $\Qh_0, \ldots, \Qh_{k-2}$ is $2nk-2k^2+k-2n +1$.
Further, stack $\Sh$ contains the following $2n-3$ edges:

\medskip
\begin{tabular}{@{~--~} l @{\hspace{0.7cm}} r @{$ \;\; \leq \; j \; \leq \;\; $} l} 
 $(v_j, v_{n-1})$               & $0$  & $k-2$ \\
 $(v_{k-2}, v_j)$               & $n - k + 1$  & $n - 2$ \\
 $(v_{(n-1)/2}, v_{j})$         & $k-2$   & $ (n-1)/2 - 4$ \\
 $(v_{(n-1)/2}, v_{j})$         & $(n-1)/2 + 4$   & $ n - k + 1 $ \\
 $(v_{j}, v_{j+1})$             & $0$ & $ n-2$ \\
\end{tabular}

\begin{tabular}{@{~--~} l @{\hspace{0.9cm}} r @{$ \;\; ~ \; j \; \in \;\; $} l} 
 $(v_{(n-1)/2 + j}, v_{(n-1)/2+j+2})$           & $~$   & $\{-2,-4,2,4\}$ \\
\end{tabular}

\medskip Hence, $G_{n,k}$ contains $2nk-2k^2 + k-2$ edges and has mixed page number $k$. To complete the proof, we observe that for each $i$ in $[1, k-1]$ no two edges in $\Qh_i$ nest, no two edges in stack $S$ cross, and no two pages share an edge. In particular, for each $j\in\{-3,-1,1,3\}$, stack $\Sh$ contains only two edges incident to $v_{(n-1)/2+j}$, while queue $\Qh_{k-2}$ contains a star at this vertex avoiding these two particular edges. Also, queue $\Qh_{k-2}$ contains two edges incident to $v_{(n-1)/2}$, while stack $\Sh$ contains a star at this vertex which avoids exactly these two edges.
\end{proof}

We conclude this section with an observation stemming from a result by Dujmovi{\'c} and Wood~\cite{DBLP:journals/dmtcs/DujmovicW04}, which states that the queue number of a graph is at most the maximum queue number of its biconnected components plus one. Since in their approach an extra queue is used (containing only edges incident to cut-vertices), the result easily caries over to the mixed page number.

\begin{observation}\label{obs:biconnected-components}
The mixed page number of a graph is upper bounded by the maximum mixed page number of its biconnected components plus one. 
\end{observation}

\section{Complete Graphs}
\label{sec:complete-graphs}

In this section, we present bounds on the mixed page number of $n$-vertex complete graphs. Notably, our upper bound is smaller than the corresponding best-known upper bounds on the stack  and on the queue number of graphs, which implies that combining stacks and queues yields linear layouts that need fewer pages in total. This observation is a follow-up on a series of observations made by Heath and Rosenberg~\cite{DBLP:journals/siamcomp/HeathR92} regarding the ``powers'' of stacks and queues. We start with our lower number on the mixed page number of $K_n$.

\begin{lemma}\label{thm:complete-lower-bound}
The mixed page number of $K_n$ is at least $\lceil\frac{3(n-4)}{8}\rceil$.
\end{lemma}
\begin{proof}
Since $K_n$ contains $\frac{1}{2}n(n-1)$ edges, \cref{thm:density} implies that its mixed page number cannot be less than $\frac{3n(n - 4)}{8 (n - 3)} \geq \frac{3n(n - 4)}{8 n}$.
\end{proof}

\begin{lemma}\label{thm:complete-upper-bound}
The mixed page number of $K_n$ is at most $2 \lceil \frac{n}{5} \rceil$.
\end{lemma}
\begin{proof}
We assume that $n$ is a multiple of $5$ and we prove that $K_n$ admits a $\frac{n}{5}$-stack $\frac{n}{5}$-queue layout $\mathcal{L}_n$; see \cref{fig:kn}.
Let $v_0,\ldots,v_{n-1}$ be the order of the vertices of $K_n$ in $\mathcal{L}_n$. 
Since the edges of $K_n$ that connect consecutive vertices in $\mathcal{L}_n$ as well as the edge $(v_0,v_{n-1})$ can be assigned to any stack of $\mathcal{L}_n$, we omit them~in~our assignment scheme.
First, we assign edges to queues $\Qh_0, \dots, \Qh_{\frac{n}{5}-1}$  of $\mathcal{L}_n$, such that for each $i$ in $[0, \frac{n}{5}-1]$, queue $\Qh_i$ contains the following $2n-4i-5$ edges:

\medskip
\begin{tabular}{@{~--~} l @{\hspace{1.0cm}} r @{$ \;\; \leq \; j \; \leq \;\; $} l} 
 $(v_i, v_j)$                  & $i + 2$                 & $\frac{4n}{5} - i - 2$ \\
 $(v_j, v_{n - 1 - i})$        & $\frac{n}{5} + 2i $     & $ n - i - 3$ \\
 $(v_{\frac{n}{5}+2i}, v_j)$   & $\frac{4n}{5} - i - 1 $ & $ n - i - 2$ \\
 $(v_j, v_{\frac{4n}{5}-i-2})$ & $i +1 $                 & $ \frac{n}{5} + 2i$
\end{tabular}
\medskip

\begin{figure}[ph!]
	\begin{subfigure}[b]{.48\linewidth}
		\flushleft
		\includegraphics[page=1,width=\textwidth]{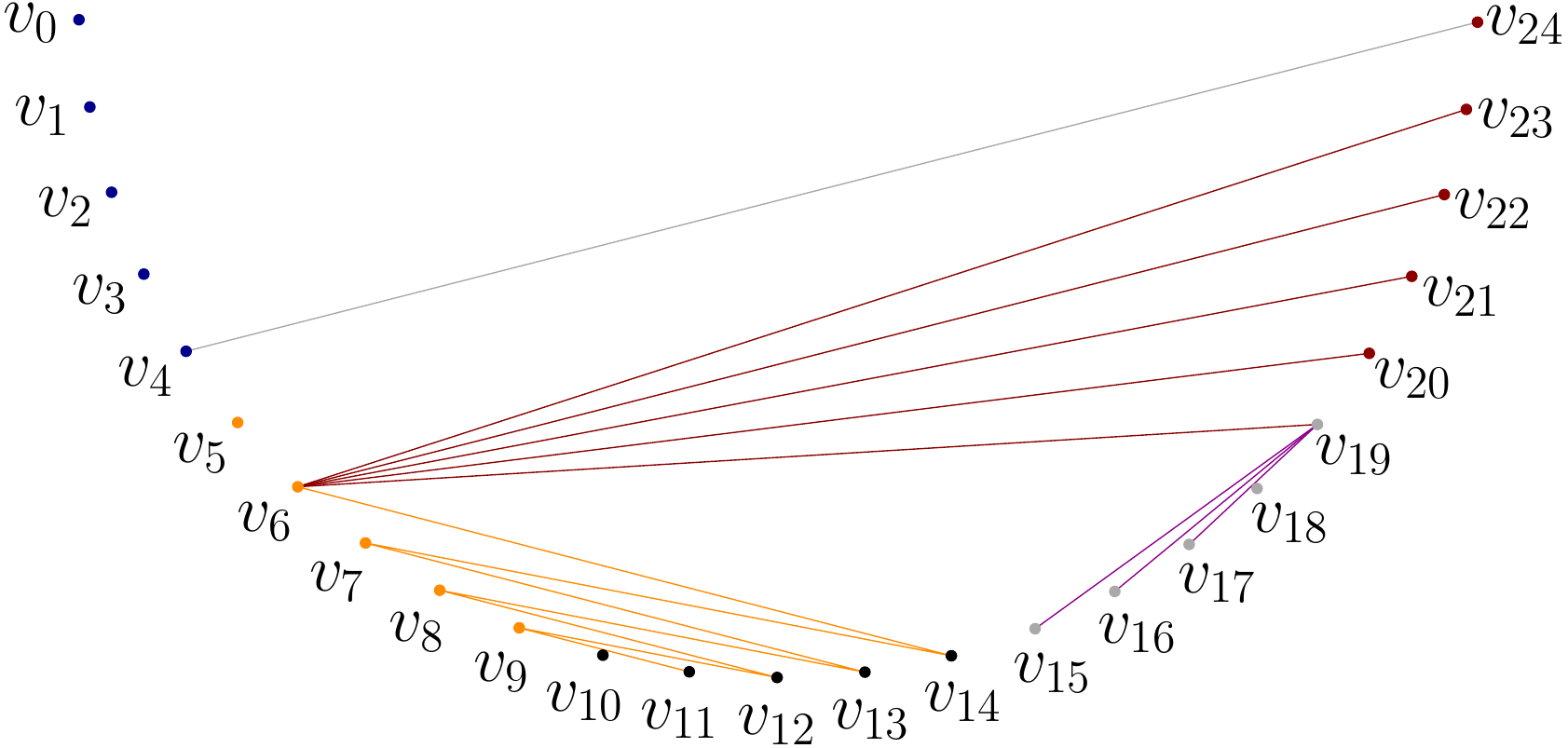}
		\caption{$\Sh_0$}
	\end{subfigure}\hfill    
	\begin{subfigure}[b]{.48\linewidth}
		\flushleft
		\includegraphics[page=2,width=\textwidth]{s-q}
		\caption{$\Sh_1$}
	\end{subfigure}
	\begin{subfigure}[b]{.48\linewidth}
		\flushright
		\includegraphics[page=3,width=\textwidth]{s-q}
		\caption{$\Sh_2$}
	\end{subfigure}\hfill
	\begin{subfigure}[b]{.48\linewidth}
		\flushleft
		\includegraphics[page=4,width=\textwidth]{s-q}
		\caption{$\Sh_3$}
	\end{subfigure}
	\begin{subfigure}[b]{.48\linewidth}
		\flushright
		\includegraphics[page=5,width=\textwidth]{s-q}
		\caption{$\Sh_4$}
	\end{subfigure}\hfill
	\begin{subfigure}[b]{.48\linewidth}
		\flushleft
		\includegraphics[page=6,width=\textwidth]{s-q}
		\caption{$\Qh_0$}
	\end{subfigure}    
	\begin{subfigure}[b]{.48\linewidth}
		\flushright
		\includegraphics[page=7,width=\textwidth]{s-q}
		\caption{$\Qh_1$}
	\end{subfigure}\hfill
	\begin{subfigure}[b]{.48\linewidth}
		\flushleft
		\includegraphics[page=8,width=\textwidth]{s-q}
		\caption{$\Qh_2$}
	\end{subfigure}
	\begin{subfigure}[b]{.48\linewidth}
		\flushright
		\includegraphics[page=9,width=\textwidth]{s-q}
		\caption{$\Qh_3$}
	\end{subfigure}\hfill
	\begin{subfigure}[b]{.48\linewidth}
		\flushleft
		\includegraphics[page=10,width=\textwidth]{s-q}
		\caption{$\Qh_4$}
	\end{subfigure}
	\caption{Illustration for the upper bound of $2\lceil\frac{n}{5}\rceil$ of \cref{thm:complete-upper-bound} with $n = 25$, which yields a $5$-stack $5$-queue layout.}
	\label{fig:kn}
\end{figure}
		
\noindent Next, we assign edges to stacks $\Sh_0, \dots, \Sh_{\frac{n}{5}-1}$ of $\mathcal{L}_n$, such that for each $i$ in $[0,\frac{n}{5}-1]$, stack $\Sh_i$ contains the following $\frac{4n}{5}+i-4$ edges (whose colors for the example of \cref{fig:kn} are indicated in the respective circles):

\definecolor{mixedred}{RGB}{139,0,0}
\definecolor{mixedgray}{RGB}{169,169,169}
\definecolor{mixedblue}{RGB}{0,0,139}
\definecolor{mixedpurple}{RGB}{139,0,139}
\definecolor{mixedorange}{RGB}{255,140,0}
\definecolor{mixedgreen}{RGB}{0,100,0}

\medskip
\begin{tabular}{ l l r @{$ \;\; \leq \; j \; \leq \;\; $} l@{~~~}l} 
\tikzcircle[mixedred, fill=mixedred]{2.8pt}       & $(v_{\frac{n}{5} + 2i + 1}, v_j) $          & $\frac{4n}{5}-i-1$  & $n - i -2 $,     &       for $(i,j)\neq(\frac{n}{5}-1,n-1)$ \\
\tikzcircle[mixedblue, fill=mixedblue]{2.8pt}     & $(v_{n - i - 1}, v_j)$                      & $\frac{n}{5}-i$     & $\frac{n}{5}-1 $  &     \\
\tikzcircle[mixedgray, fill=mixedgray]{2.8pt}     & $(v_{\frac{n}{5} - i - 1}, v_j)$            & $n - i - 1$         & $n - 1$,           &  for $i\neq (\frac{n}{5}-1)$ \\
\tikzcircle[mixedpurple, fill=mixedpurple]{2.8pt} & $(v_{\frac{4n}{5} - i - 1}, v_j)$           & $\frac{3n}{5}$      & $\frac{4n}{5} - i - 3$,& for $(i,j)\neq(\frac{n}{5}-1,\frac{n}{5}-1)$ \\
\tikzcircle[mixedorange, fill=mixedorange]{2.8pt} & $(v_{2i + j + 1}, v_{\frac{2n}{5} - j -1})$ & $\frac{n}{5} $     & $\frac{2n}{5} - i - 2$&  \\
\tikzcircle[mixedorange, fill=mixedorange]{2.8pt} & $(v_{2i + j + 2}, v_{\frac{2n}{5} - j -1})$ & $\frac{n}{5} $     & $\frac{2n}{5} - i - 3$,& for $(i,j)\neq(\frac{n}{5}-1,\frac{n}{5}-1)$ \\
\tikzcircle[mixedgreen, fill=mixedgreen]{2.8pt}   & $(v_{j}, v_{2i - j + 1})$                   & $\frac{n}{5} $     & $\frac{n}{5} + i-1$  &    \\
\tikzcircle[mixedgreen, fill=mixedgreen]{2.8pt}   & $(v_{j}, v_{2i - j})$                       & $\frac{n}{5} $     & $\frac{n}{5} + i-1$  &    \\
\end{tabular}
\medskip

\noindent By construction, no two edges in the same stack (queue) cross (nest, respectively), and no edge is assigned to two distinct pages (i.e., stacks or queues). To complete the proof, we count the total number of edges in $\mathcal{L}_n$. Summing up~over $i$, we conclude that $\mathcal{L}_n$ contains $\frac{8n^2}{25} - \frac{3n}{5}$ and $\frac{9n^2}{50} - \frac{9n}{10}$ edges in its queues and stacks, respectively. Thus, in total $\mathcal{L}_n$ has $\frac{1}{2}(n^2-3n)$ edges, i.e., the number of edges of $K_n$ neglecting the edges connecting consecutive vertices, as desired.
\end{proof}	

\section{Complete Bipartite Graphs}
\label{sec:complete-bipartite-graphs}

In this section, we study bounds on the mixed page number of the complete bipartite graph $K_{n,n}$, i.e., we assume that both parts of it are of the same cardinality. We start with the natural case, in which the vertices of one part of $K_{n,n}$ precede those of its second part; we refer to this scenario as the \emph{separated setting} (\cref{ssec:separated}). Then, we analyze the non-separated case in which no restriction on precedence of the parts is imposed (\cref{ssec:non-separated}). 

\subsection{The separated setting}
\label{ssec:separated}

Let $u_0,\ldots,u_{n-1}$ and $v_0,\ldots,v_{n-1}$ be the vertices of the two parts of $K_{n,n}$. In the separated setting, we assume w.l.o.g.\ $u_0 \prec \ldots \prec u_{n-1} \prec v_0 \prec \ldots \prec v_{n-1}$; see \cref{fig:mixed-k66}. 
Since by bipartiteness no two edges can form a necklace, any two edges either share an endpoint, cross, or nest. Furthermore, two crossing (nesting) edges are nesting (crossing), when one reverses the order of the vertices of one of its two parts. Hence, reversing the order of the vertices of one part of an $s$-stack $q$-queue layout yields a $q$-stack $s$-queue layout. In this sense, stacks and queues are \emph{interchangeable}. This property will allow us to~translate properties proved for queues to corresponding properties for stacks, and vice~versa.  

\begin{figure}[!h]
    \begin{minipage}[b]{0.58\textwidth}
    	\begin{subfigure}[b]{.48\linewidth}
    		\flushleft
    		\includegraphics[page=1,scale=0.82]{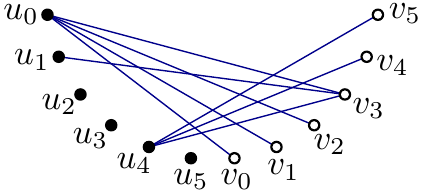}
    		\caption{$\Qh_0$}
    	\end{subfigure}\hfill
    	\begin{subfigure}[b]{.48\linewidth}
    		\flushright
    		\includegraphics[page=2,scale=0.82]{Knn_sep}
    		\caption{$\Qh_1$}
    	\end{subfigure}    
    	\begin{subfigure}[b]{.48\linewidth}
    		\flushleft
    		\includegraphics[page=3,scale=0.82]{Knn_sep}
    		\caption{$\Sh_0$}
    	\end{subfigure}\hfill
    	\begin{subfigure}[b]{.48\linewidth}
    		\flushright
    		\includegraphics[page=4,scale=0.82]{Knn_sep}
    		\caption{$\Sh_1$}
    	\end{subfigure}
    \end{minipage}\hfill
    \begin{minipage}[b]{0.38\textwidth}
        \centering
    	\begin{subfigure}[b]{.98\linewidth}
    		\flushright
    		\includegraphics[page=5,scale=0.82]{Knn_sep}
    		\caption{$H$}
            \label{fig:grid}
    	\end{subfigure}
    \end{minipage}
	\caption{(a)--(d)~A $2$-stack $2$-queue layout of $K_{6,6}$ in the separated setting, and (e)~the grid corresponding to this layout.}
	\label{fig:mixed-k66}
\end{figure}

The key ingredient of our approach is a one-to-one mapping of the $n^2$ edges of~$K_{n,n}$ to the $n^2$ points of the $n \times n$ integer grid $\grid = [0,n-1]\times[0,n-1]$, such that the edge $(u_i,v_j)$ of~$K_{n,n}$ is mapped to point $(i,j)$ of $\grid$; see \cref{fig:grid}. Observe that the interchangeability property just mentioned corresponds to mirroring $H$ either along the $x$-axis or along the $y$-axis, depending on the part being reversed.
In the following, we introduce an important property of this mapping. 

\begin{proposition}\label{prp:sep-matrix}
Let $\mathcal{L}$ be an $s$-stack $q$-queue layout of $K_{n,n}$ in the separated setting. Then, the edges assigned to the same queue (stack) of $\mathcal{L}$ form a not necessarily strict monotonically increasing (decreasing, respectively) path on $\grid$.
\end{proposition}
\begin{proof}
We prove the statement for the edges assigned to the same queue of $\mathcal{L}$; by the interchangeability of stacks and queues in the separated setting, a symmetric argument applies to the edges of the same stack of $\mathcal{L}$. We consider the edges assigned to the same queue in their lexicographic order, namely 
edge $(u_i,v_j)$ \emph{precedes} another edge $(u_k,v_\ell)$ if and only if $i=k$ and $j < \ell$ holds, or $i < k$ and $j \leq \ell$ holds. Thus, in this order, we can observe that for any edge $(u_i,v_j)$ any subsequent edge in the lexicographic order $(u_k,v_\ell)$ is mapped to a point $(k,\ell)$ in the upper-right quadrant of $(i,j)$ in $\grid$ (see \cref{fig:grid}), i.e., the edges assigned to the same queue of $\mathcal{L}$ form a monotonically increasing path of $\grid$. 
\end{proof}

\noindent In view of \cref{prp:sep-matrix}, we refer in the following to a monotonically decreasing (increasing) path of $\grid$ as a \emph{stack-path} (\emph{queue-path}, respectively).

\begin{proposition}
A stack-path starting at $(i,j)$ and ending at $(k,\ell)$ covers~at~most $k-i + j - \ell + 1$ grid points; a corresponding queue-path at most $k-i + \ell - j + 1$.
\end{proposition}
\begin{proof}
Directly follows from monotonicity and from the fact that the Manhattan distance between $(i,j)$ and $(k,\ell)$ is $|k-i|+|\ell-j|$.
\end{proof}

\noindent With the property of \cref{prp:sep-matrix} in mind, we are now ready to introduce our upper bound on the mixed page number of $K_{n,n}$ in the separating setting.

\begin{lemma}\label{lem:sep-upper-bound}
The mixed separated number of $K_{n, n}$ is at most $\lceil\frac{2n}{3}\rceil$.
\end{lemma}
\begin{proof}
Assuming that $n$ is a multiple of $3$, we prove that $K_{n, n}$ admits a $\frac{n}{3}$-stack $\frac{n}{3}$-queue layout $\mathcal{L}_n$. 
By \cref{prp:sep-matrix}, this is equivalent to determining how all $n^2$ points of the integer grid $H$ can be covered with a set $Q$ of $\frac{n}{3}$ queue-paths and a set $S$ of $\frac{n}{3}$ stack-paths; see \cref{fig:windmill}.
For $0 \leq i \leq \frac{n}{3}-1$, the $i$-th path of $S$ covers the following $\frac{5n}{3}$ points of $\grid$: 
$(0,n-1-i), \ldots, (\frac{2n}{3}-1-i,n-1-i), \ldots, (\frac{2n}{3}-1-i,\frac{n}{3}-1-i), \ldots, (n-1,\frac{n}{3}-1-i)$.
For each $0 \leq i \leq \frac{n}{3}-1$, the $i$-th path of $Q$ covers the following $\frac{5n}{3}$ points of $\grid$: 
$(i,0), \ldots, (i,\frac{2n}{3}-1-i),  \ldots,(\frac{2n}{3}+i,\frac{2n}{3}-1-i), \ldots, (\frac{2n}{3}+i,n-1)$.
So, in total the paths in~$S$~and~$Q$ cover $\frac{10n^2}{9}$ points, and since every pair of a path in $S$ and a path in $Q$ has exactly one point of $\grid$ in common, $\frac{n^2}{9}$ points have been double-counted. Thus, all $n^2$ points of $\grid$ have been covered, which implies that all edges of $K_{n,n}$ are present in $\mathcal{L}_n$. 

If $n$ is not a multiple of~$3$, we first compute a $\frac{n'}{3}$-stack $\frac{n'}{3}$-queue layout $\mathcal{L}_{n'}$ for $K_{n',n'}$ with $n'=3\lfloor\frac{n}{3}\rfloor$. We augment $\mathcal{L}_{n'}$ to a layout of $K_{n,n}$ using $n\mod 3$ additional stacks or queues, which completes the proof. 
\end{proof}

\begin{figure}[!tb]
	\begin{subfigure}[b]{.48\linewidth}
		\centering
		\includegraphics[page=8,scale=1.1]{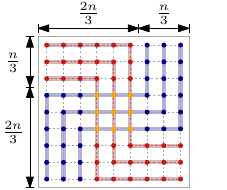}
		\caption{}
		\label{fig:windmill}
	\end{subfigure}\hfill
	\begin{subfigure}[b]{.48\linewidth}
		\centering
		\includegraphics[page=9, scale=1.1]{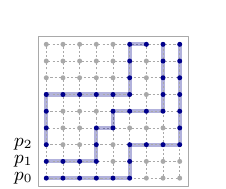}
		\caption{}
		\label{fig:queue-paths}
	\end{subfigure}  
	\caption{Illustrations for the proofs of 
	(a)~\cref{lem:sep-upper-bound} in which every stack-path shares exactly one point with every queue-path, and 
	(b)~\cref{prp:queue-paths} in which the queue-paths $p_0\ldots,p_{q-1}$ cover a largest set in $H$.}
	\label{fig:sep-bounds}
\end{figure}

\noindent Next, we focus on a corresponding lower bound for $K_{n, n}$. For this, we need~one more ingredient that we formally prove for  queue-paths and by the interchangeability of stacks and queues translates also to stack-paths (\cref{prp:queue-paths,prp:stack-paths}). 

\begin{proposition}\label{prp:queue-paths}
Let $Q$ be a set of $q$ queue-paths covering a largest set $P(Q)$ of points of $\grid$. 
Then, $P(Q)$ can be covered by $q$ queue-paths $p_0\ldots,p_{q-1}$, where~$p_i$ has length $2n-1-2i$, starts at $(0,i)$ and ends at $(n-1-i,n-1)$; $i = 0, \ldots, q-1$.
\end{proposition}
\begin{proof}
Since each path in $Q$ is a queue-path, by \cref{prp:sep-matrix} the boundary of $P(Q)$ is formed by two sets of points, each of which is monotonically increasing. Assume that, for some $i$ in $\{0, \ldots,q-1\}$, queue-paths $p_0,\ldots,p_{i-1}$ have been computed, such that for each $j$ in $\{0,\ldots,i-1\}$ queue-path $p_j$ has length $2n-1-2j$, starts at $(0,j)$, ends at $(n-1-j,n-1)$, and $p_0,\ldots,p_{i-1}$ cover all points of $P(Q)$ whose $x$- and $y$-coordinates are in $\{n-i,\ldots,n-1\}$ and in $\{0,\ldots,i-1\}$, respectively.
Since each of the paths $p_0,\ldots,p_{i-1}$ is a queue-path, the boundary of the set $P(Q_i)$ of points in $P(Q)$ that have not been covered by $p_0,\ldots,p_{i-1}$ is formed by two sets of points, each of which is monotonically~increasing. Also, by assumption, $P(Q_i)$ does not contain any point whose $x$-coordinate is in $\{n-i,\ldots,n-1\}$ or whose $y$-coordinate is in $\{0,\ldots,i-1\}$.
On the other hand, $P(Q_i)$ necessarily contains at least one point whose $x$-coordinate is $n-1-i$ and one whose $y$-coordinate is $i$, 
as otherwise a path of $Q$ can be elongated to contain such point contradicting the fact that $Q$ covers a largest set of points of $\grid$. For $j=i,\ldots,n-i-1$, the next queue-path $p_i$ contains the rightmost point of $P(Q_i)$ whose $y$-coordinate is $j$.
This uniquely defines a queue-path of length $2n-1-2i$ starting at $(0,i)$ and ending at $(n-1-i,n-1)$ that covers all points of $P(Q)$ whose $x$-coordinate is $n-1-i$ and whose $y$-coordinate is $i$; see \cref{fig:queue-paths}.

We next argue that $p_0,\ldots,p_{q-1}$ are pairwise disjoint and cover all points of $P(Q)$. Since $\sum_{i=0}^{q-1}2n-1-2i$ is an upper bound on the cardinality of $P(Q)$, if we give a point set of $\grid$ that achieves this bound, then we have proved that any maximal point set $P(Q)$ must have this cardinality. 
Indeed, the following set of queue-paths are pairwise-disjoint and cover exactly $\sum_{i=0}^{q-1}2n-1-2i$ points of~$\grid$: 
for each $0 \leq i \leq q-1$, the $i$-th path covers points $(i,0), \ldots, (i,n-1-i), \ldots,(n-1,n-1-i)$, i.e., $2n-1-2i$ points. 
This implies that in any maximal point set (including $P(Q)$), which is covered by a set of paths of these specific lengths, the paths must be pairwise disjoint, as we claimed. This further implies that even if the original $q$ queue-paths that cover $P(Q)$ are not pairwise disjoint, there exist subpaths $p'_0,\ldots,p'_{q-1}$ of them that are necessarily pairwise disjoint and still cover $P(Q)$. W.l.o.g., we assume that $p'_0,\ldots,p'_{q-1}$ are ordered from the lowest one to the highest one in the boundary of $P(Q)$. Hence, our construction guarantees that for each $i=0,\ldots,q-1$, path $p_i$ will completely contain $p'_i$. Hence, all points of $P(Q)$ are covered by $p_0,\ldots,p_{q-1}$, as claimed.
\end{proof}

\noindent Symmetrically, the following proposition can be proven. 

\begin{proposition}\label{prp:stack-paths}
Let $S$ be a set of $s$ stack-paths covering a largest set $P(S)$ of points of $\grid$. 
Then, $P(S)$ can be covered by $s$ stack-paths $p_0\ldots,p_{s-1}$, where~$p_i$ has length $2n-1-2i$, starts at $(0,n-1-i)$ and ends at $(n-1,i)$; $i = 0, \ldots, s-1$.
\end{proposition}

\noindent The next lemma provides an estimation on the maximum number edges of an $s$-stack $q$-queue layout of $K_{n,n}$ in the separated setting. 

\begin{lemma}
Let $S$ and $Q$ be two sets of $s$ stack-paths and $q$ queue-paths covering a largest set of points of $\grid$. Then, $S \cup Q$ covers $\sum_{i=0}^{s-1} (2n - 1 -2i) + \sum_{i=0}^{q-1} (2n - 1 -2i) - sq$ grid points of $\grid$.
\end{lemma}
\begin{proof}
W.l.o.g., we assume that $S$ and $Q$ are maximal, as otherwise we can simply extend them to maximal. Then, the sets $P(S)$ and $P(Q)$ of the points~of~$\grid$ that are covered by the paths in $S$ and $Q$, respectively, can be covered by $s$ stack-paths and $q$ queue-paths that have the properties of the last two propositions. Clearly, each such stack-path and each such queue-path have at least one point of $\grid$ in common. Therefore, $|P(S) \cap P(Q)| \geq sq$ holds, and thus we obtain:  
%
\begin{equation*}
\begin{split}
|P(S) \cup P(Q)| & \leq |P(S)| + |P(Q)| - |P(S) \cap P(Q)|   \\
                 & = \sum_{i=0}^{s-1} (2n - 1 -2i) + \sum_{i=0}^{q-1} (2n - 1 -2i) - sq\qedhere
\end{split}
\end{equation*}
\end{proof}

\begin{corollary}\label{cor:bipartite-density}
An $s$-stack $q$-queue layout of $K_{n, n}$ in the separated case contains at most $2n(q+s) - q^2 - s^2 - sq$ edges.
\end{corollary}

\noindent Next, we prove that the upper bound of \cref{lem:sep-upper-bound} is tight. 
 
\begin{theorem}\label{lem:sep-bound}
The mixed separated number of $K_{n, n}$ is $\lceil \frac{2n}{3} \rceil$.
\end{theorem}
\begin{proof}
The upper bound follows from \cref{lem:sep-upper-bound}. 
For the lower bound, denote~by $k$ the total number of pages in an $s$-stack, $q$-queue layout $\mathcal L_n$ of $K_{n, n}$. 
By \cref{cor:bipartite-density}, we obtain
$2n(q+s) - q^2 - s^2 - sq \geq n^2$.
Since $k=s+q$, it follows
$2 n k - q^2 - k^2 + kq \geq n^2$. 
To determine the maximum for the left-hand side of this inequality with respect to $q$, we compute the roots of its first derivative taken over $q$, which is 
$\frac{\partial}{\partial q}\left( 2 n k - q^2 - k^2 + kq\right) = 0$ and yields $k = 2q$,
implying $s=q=\frac{k}{2}$. Thus, 
$ 2 n k - \frac{3k^2}{4} - n^2 \geq 0$ holds.
Hence, $k \geq \lceil \frac{2n}{3} \rceil$, as claimed.
\end{proof}	

\subsection{Non-separated setting}
\label{ssec:non-separated}

We now study the mixed page number of $K_{n, n}$ in the general (i.e., non-separated) setting. For this setting, we are not able to provide an upper bound that is better than $\lfloor \frac{n}{2} \rfloor$, i.e., the queue number of $K_{n, n}$. We conjecture~that this bound is tight, but we do not have a proof for this. However, we are able to provide a lower bound using using the~techniques of the previous section.

\begin{theorem}\label{thm:bipartite-pages}
The mixed page number of $K_{n, n}$ is at least $\lceil \frac{n}{3} \rceil$.
\end{theorem}

\begin{proof}
By the pigeonhole principle, in any vertex order of $K_{n, n}$ at least $\lceil \frac{n}{2} \rceil$~of the first $n$ vertices belong to one part. Then at least $\lceil \frac{n}{2} \rceil$ of the remaining $n$ vertices belong to the other part. So, every vertex order induces a $K_{\lceil\frac{n}{2}\rceil,\lceil\frac{n}{2}\rceil}$ in the separated setting, which by \cref{lem:sep-bound} requires at least $\lceil \frac{n}{3} \rceil$~pages.
\end{proof}

\section{Conclusions}
\label{sec:conclusions}

In this work, we shed some light on the mixed page number of some basic graph classes. We want to emphasize the following open questions:
\begin{itemize}[--]
    \item Can the gaps of the bounds on the mixed page numbers of $K_n$ and $K_{n,n}$~be closed? We believe that space for improvement is on the side of the lower bounds.
    \item Do graphs with bounded mixed page number have bounded stack or queue number?
    By a recent result of Dujmovi{\'c} et al.~\cite{DEHMW20}, the answer can be affirmative only if queue number is bounded by stack number for every graph.
    The question is interesting even for separated layouts.
    \item Since planar graphs have stack number $4$~\cite{DBLP:journals/jocg/KaufmannBKPRU20,DBLP:journals/jcss/Yannakakis89} (while their queue number ranges between $4$ and $49$~\cite{DBLP:journals/algorithmica/AlamBGKP20,DBLP:journals/jacm/DujmovicJMMUW20}), we ask whether their mixed page number is $3$ or $4$.
    \item Finally, we suggest to investigate the mixed page number of other classes of graphs such as $k$-planar graphs.
\end{itemize}

\bibliographystyle{splncs03}
\bibliography{general,stacks,queues}

\end{document}